\documentclass{article}

\usepackage{amssymb}
\usepackage{amsmath}
\usepackage{latexsym}
\usepackage[dvips]{color}
\usepackage{epsfig}

\newtheorem{definition}{Definition}

\newtheorem{theorem}{Theorem}
\newtheorem{proof}{Proof}

\begin{document}
\setcounter{page}{1}

\date{Written 2010}

\title{General upper bounds for well-behaving goodness measures on dependency rules}

\author{Wilhelmiina H{\"a}m{\"a}l{\"a}inen\\
Department of Computer Science\\ 
P.O. Box 68, FI-00014 University of Helsinki, Finland\\
whamalai{@}cs.uef.fi} 
\maketitle

\begin{abstract}
In the search for statistical dependency rules, a crucial task is to
restrict the search space by estimating upper bounds for the goodness
of yet undiscovered rules. In this paper, we show that all
well-behaving goodness measures achieve their maximal values in the
same points. Therefore, the same generic search strategy can be
applied with any of these measures. The notion of well-behaving
measures is based on the classical axioms for any proper goodness
measures, and extended to negative dependencies, as well. As an example, 
we show that several commonly used goodness measures are well-behaving. 
\end{abstract}


\section{Introduction}

In the rule discovery, a general task is to search rules of form
$X\rightarrow A$, where $X$ is a set of true-valued binary attributes,
$A$ is a binary attribute, and $X$ and $A$ are statistically
dependent. In practice, the problem occurs in two forms: one may
either want to enumerate all sufficiently good rules or to search the
best $K$ rules. In both cases, the goodness of a rule $X \rightarrow
A$, is estimated by some goodness measure $M$. In the enumeration
problem, one should find all rules $X\rightarrow A$, for which
$M(X\rightarrow A)\geq min_M$, for some minimum threshold $min_M$. If
small $M$ values indicate a good rule, then a threshold $max_M$ is used
instead. In the optimization problem, one should find $K$ rules which
have maximal (or minimal) $M$ values among all possible rules.

In both search tasks a crucial problem is to restrict the search space
by estimating a tight upper bound (or a lower bound) for $M(XQ
\rightarrow A)$ for any rule $XQ \rightarrow A$, when a more general
rule $X\rightarrow A$ is known. Based on this upper bound $ub$, one
can prune the rule $XQ \rightarrow A$ without further checking, if
$ub$ is too small.

In this research note, we prove general upper/lower bounds, which hold
for any {\em well-behaving} measure $M$. The notion of well-behaving
is defined by the classical axioms introduced in
\cite{piatetskyshapiro} and \cite{major}. In practice, the axioms hold
for a large class of popular goodness measures used for evaluating the
goodness of statistical dependencies, classification rules, or
association rules.

\begin{table}
\caption{Basic notations.}
\label{notations}
\begin{center}
\begin{tabular}{ll}
$A$&single attribute\\
$X,Q,R$& attribute sets\\
$m(X)$&absolute frequency of $X$\\
$P(X)$&relative frequency of $X$\\
$cf$, $P(A|X)$&confidence  $P(A|X)=\frac{P(XA)}{P(X)}$\\
$\delta$&leverage, $P(XA)-P(X)P(A)$\\
$N_X,N_{XA},N_A,N\in \mathbb{N}$&random variables\\
\end{tabular}
\end{center}
\end{table}

For any rule $X\rightarrow A$, the measure value $M(X\rightarrow A)$
can be determined as a function of four variables: absolute
frequencies $N_X=m(X)$, $N_{XA}=m(XA)$, $N_A=m(A)$, and data size
$N=n$ (see basic notations in Table \ref{notations}). Let us now
assume that the measure $M$ is increasing by goodness, meaning that
high values of $M(X\rightarrow A)$ indicate that $X\rightarrow A$ is a
good rule. According to classical axioms by Piatetsy-Shapiro
\cite{piatetskyshapiro} the following axioms should hold for any
proper measure $M$, measuring the goodness of a positive dependency
between $X$ and $A$:

\begin{itemize}
\item[(i)] $M$ minimal, when $NN_{XA}=N_XN_A$,  
\item[(ii)] $M$ is monotonically increasing with $N_{XA}$, when $N_X$,  
$N_A$, and $N$ remain unchanged, and 
\item[(iii)] $M$ is monotonically decreasing with $N_X$ (or $N_A$), when 
$N_{XA}$, $N_A$ (or $N_X$), and $N$ remain unchanged. 
\end{itemize}

The first axiom simply states that $M(X\rightarrow A)$ gets its
minimum value, when $X$ and $A$ are independent. In addition, it is
(implicitly) assumed that $M$ gets the minimum value also for
negative dependencies, i.e.\ when $NN_{XA}<N_XN_A$. The second axiom
states that $M$ increases, when the dependency becomes stronger 
(leverage $\delta=P(XA)-P(X)P(A)$ increases) and the rule becomes
more frequent. The third axiom states that $M$ decreases, when the
dependency becomes weaker ($\delta$ decreases). 

In \cite{piatetskyshapiro} it was noticed that under these conditions
$M$ gets its maximal value for any fixed $m(X)$, when $m(XA)=m(X)$. In
addition, it was assumed that $M$ would get its global maximum
(supremum), when $m(XA)=m(X)=m(A)$. However, the latter does not
necessarily hold, because the axioms do not tell how to compare rules
$X\rightarrow A$ and $XQ\rightarrow A$, when $P(A|X)=P(A|XQ)=1$. In
this case, the more general rule, $X\rightarrow A$, has both larger
$N_{XA}$ and $N_X$ than $XQ\rightarrow A$ has. Major and Mangano 
\cite{major} suggested a fourth axiom, which solves this
problem:

\begin{itemize}
\item[(iv)] $M$ is monotonically increasing with $N_X$, when
  $cf=\frac{N_{XA}}{N_X}$ is fixed, $N_A$ and $N$ are fixed, and
  $cf>\frac{N_A}{N}$.
\end{itemize}

According to this axiom, a more general rule is better, when two rules
have the same (or equally frequent) consequent and the confidence is
the same. In addition, it was required that the dependency should be
positive. However, based on our derivations, we assume that the same
property holds also for negative dependencies, and $M$ is
non-increasing only when there is independence. In Appendix
\ref{appA}, we show this for the $\chi^2$-measure and mutual
information.

In the following, we extend the axioms for negative dependencies and
prove general upper bounds which hold for any measure $M$ following
these axioms. The upper bounds are the same as derived in
\cite{morishitasese} in the case of the $\chi^2$-measure. In
\cite{morishitasese}, the upper bounds were derived by showing that
the $\chi^2$ is a convex function of $N_X$ and $N_{XA}$ for a fixed
consequent $A$. Similar results could be achieved for other convex
measures, but checking the axioms is simpler than convexity proofs. In
addition, there are non-convex goodness measures, which still follow
the axioms (an example of a non-convex and non-concave well-behaving
measure is the $z$-score $z_1$, which we consider in Appendix
\ref{appA}).

\section{Goodness measures for dependency rules}

Let us first define a general goodness measure for dependency
rules. For simplicity, we consider rules of form $X\rightarrow A=a$, 
where $a\in\{0,1\}$, i.e. rules $X\rightarrow A$ and $X \rightarrow \neg A$. 

\begin{definition}[Goodness measure]
Let $R$ be a set of binary attributes and ${\mathcal{U}}=\{X \rightarrow
A=a~|~X\subsetneq R, A\in R\setminus X, a\in\{0,1\}\}$ the set of all
possible rules, which can be constructed from attributes $R$. 

Let $f(N_X,N_{XA},N_A,N): \mathbb{N}^4 \rightarrow \mathbb{R}$ be some
statistical measure function, which measures the significance of
positive dependency between $X$ and $A=a$, given absolute frequencies 
$N_X=m(X)$, $N_{XA}=m(XA=a)$, $N_A=m(A=a)$, and data size $N=n$.

Function $M:{\mathcal{U}}\rightarrow \mathbb{R}$ is a goodness measure
for dependency rules, if $M(X\rightarrow A=a)=f(m(X),m(XA=a),m(A=a),n)$.
 
Measure $M$ is called increasing (by goodness), if large values of
$M(X\rightarrow A=a)$ indicate that $X\rightarrow A=a$ is a good rule,
and, respectively, decreasing, if low values indicate a good rule.
\end{definition}

In the above definition, we have defined the statistical measure
function on parameters $N_X$, $N_{XA}$, $N_A$, and $N$. First, we note
that in practice, some of these parameters can be considered
constant. For example, if the data size $N=n$ is given, it can be
omitted. If the consequent $A=a$ is also fixed, we can use a simpler
function $f_{A=a}(N_X,N_{XA})$. 

Second, we note that even if we have defined the function $f$ in
whole $\mathbb{N}^4$, only some parameter value combinations can occur
in any real data set. For any real frequencies hold $0\leq n$, $0\leq
m(X)\leq n$, $0\leq m(A=a)\leq n$, and $0\leq m(XA=a)\leq
\min\{m(X),m(A=a)\}$. In addition, for any non-trivial rule must hold
$0<n$, $0<m(X)<n$, and $0<m(A=a)<n$.  If $n=0$, the data set would not
exist. If $m(A=a)$ or $m(X)$ were 0, the corresponding rule
$X\rightarrow A=a$ would not occur in the data at all. On the other
hand, if $m(X)=n$ or $m(A=a)=n$, then either $X$ or $A=a$ would occur on
all rows of data, and the rule could express only independence.
Therefore, it suffices that the function $f$ is defined in the set of all
legal parameter values.

Third, we note that the actual function can be defined on other
parameters, if they can be derived from $N_X$, $N_{XA}$, $N_A$, and
$N$. Examples of commonly occurring derived parameters are leverage
 $\delta$ and confidence $cf$ (Table \ref{notations}). For example, when the
data size $n$ and the consequent $A$ are fixed, the $\chi^2$-measure
can be defined e.g.\ by the following two functions:
$$f_1(N_X,N_{XA})=\frac{n(N_{XA}-N_XP(A))^2}{N_X(n-N_X)P(A)(1-P(A))}$$ 
$$f_2(N_X,\delta)=\frac{n^3\delta^2}{N_X(n-N_X)P(A)(1-P(A))}.$$
Functions $f_1$ and $f_2$ can be transformed to each other by equalities
$$f_1(N_X,N_{XA})=f_2(N_{X},\frac{N_{XA}-N_XP(A)}{N})$$
$$f_2(N_X,\delta)=f_1(N_X,N_XP(A)+n\delta).$$
These transformations are often useful, when the behaviour of the
function analyzed.

Next, we will define a general class of measure functions, for which
upper (or lower) bounds can be easily provided. The criteria for such
well-behaving goodness measures are based on the classical axioms. For
simplicity, we assume that $M$ is increasing by goodness; for a decreasing
measure, the properties are reversed (minimum vs. maximum, increasing
vs. decreasing).

\begin{definition}[Well-behaving goodness measure]
Let $M$ be an increasing goodness measure defined by
function $f(N_X,N_{XA},N_A,N)$. Let $f_2(N_X,\delta,N_A,N)$ 
be another function which defines the same measure:\\
$f_2(N_X,\delta,N_A,N)=f(N_X,\frac{N_XN_A}{N}+\delta N,N_A,N)$ and
$f(N_X,N_{XA},N_A,N)=f_2(N_X,\frac{NN_{XA}-N_XN_A}{N^2},N_A,N)$. 

Let $S \subseteq \mathbb{N}^4$ be a set of all legal parameter 
values $(N_X,N_{XA},N_A,N)$ for an arbitrary data set. 

Measure $M$ is called well-behaving, if it has the following 
properties in set $S$: 
\begin{itemize}
\item[(i)] $f_2$ gets its minimum value, when $\delta=0$.
\item[(ii)] If $N_X$, $N_A$, and $N$ are 
fixed, then 
\begin{itemize}
\item[(a)] $f_2$ is a monotonically increasing function of 
$\delta$, when $\delta> 0$ (positive dependence), and 
\item[(b)] $f_2$ is a monotonically  
decreasing function of $\delta$, when $\delta<0$ (negative 
dependence).
\end{itemize}
\item[(iii)] If $N_{XA}=m(XA=a)$, $N_A=m(A=a)$, and $N=n$ are fixed, then 
\begin{itemize}
\item[(a)] $f$ is a monotonically decreasing function of 
$N_X$, when $N_X< \frac{n\cdot m(XA=a)}{m(A=a)}$ (positive dependence), and 
\item[(b)] $f$ is a monotonically increasing function of $N_X$, when 
$N_X> \frac{n\cdot m(XA=a)}{m(A=a)}$ (negative dependence).
\end{itemize}
\item[(iv)] If $N_A=m(A=a)$ and $N=n$ are fixed, then for all 
$cf_1,cf_2\in[0,1]$ 
\begin{itemize}
\item[(a)] $f(N_X,cf_1N_X,m(A=a),n)$ is monotonically increasing with $N_x$, 
when $cf_1>\frac{m(A=a)}{n}$ (positive dependence), and
\item[(b)]  $f(N_X,m(A=a)-cf_2(n-N_X),m(A=a),n)$ is monotonically 
decreasing with $N_X$, when $cf_2>\frac{m(A=a)}{n}$ (negative dependence).
\end{itemize}
\end{itemize}
\end{definition}

The first two conditions are obviously equivalent to the classical
axioms (i) and (ii). The only difference is that the behaviour is
expressed in terms of leverage $\delta$. This enables that the measure
can be defined for negative dependencies, as well. In addition, it is
often easier to check that the conditions hold for a desired function,
when it is expressed as a function of $\delta$. In practice, this can
be done by differentiating $f_2(N_X,\delta,N_A,N)$ with respect to
$\delta$, where $N_X$, $N_A$, and $N$ are considered constants. If $M$
is an increasing function, then the derivative
$f_2'=f_2'(N_X,\delta,N_A,N)$ should be $f_2'=0$, when $\delta=0$,
$f_2'>0$, when $\delta>0$, and $f_2'<0$, when $\delta<0$. For
decreasing $M$, the signs of the derivative are reversed. We note that
it is enough that the function is defined and differentiable in the set of
legal values.

The third condition is also equivalent to the classical axiom (iii),
when extended to both positive and negative dependencies. Before
point $N_X=\frac{n\cdot m(XA=a)}{m(A=a)}$, $M$ measures positive
dependence, and, after it, negative dependence. The third condition can
be checked by differentiating $f(N_X,m(XA=a),m(A=a),n)$ with respect to $N_X$,
where $m(XA=a)$, $m(A=a)$, and $n$ are constants. For an increasing $M$,
the derivative $f'=f'(N_X,m(XA=a),m(A=a),n)$ should be $f'=0$, when
$M_X=\frac{n\cdot m(XA=a)}{m(A=a)}$ (independence), $f'<0$, when
$N_X<\frac{n\cdot m(XA=a)}{m(A=a)}$ (positive dependence), and $f'>0$, 
when $N_X>\frac{n\cdot m(XA=a)}{m(A=a)}$ (negative dependence). For 
decreasing $M$, the signs are reversed.

Similarly, the fourth condition is equivalent to the classical axiom
(iv), when extended to both positive and negative dependencies. In the
case of positive dependence, $cf_1$ corresponds to confidence
$P(A=a|X)$, which is kept fixed. In the case of negative dependence,
$cf_2$ corresponds to confidence $P(A=a|\neg X)$, which is kept
fixed. (Now the equation is $N_{XA}=m(A=a)-cf_2(n-N_X) \Leftrightarrow
m(A=a)-N_{XA}=cf_2(n-N_X) \Leftrightarrow N_{\neg XA}=cf_2N_{\neg
  X}$.) We require that condition (a) holds for positive dependence
($P(A=a|X)>P(A)$) and (b) for negative dependence ($P(A=a|\neg
X)>P(A=a) \Leftrightarrow P(A=a|X)<P(A=a)$). However, according to our
analysis of the $\chi^2$-measure and mutual information (Appendix \ref{appA}),
both (a) and (b) hold everywhere, where $P(A=a|X)\neq P(A=a)$ and
$P(A=a|\neg X)\neq P(A=a)$. It is still unproved, whether this holds
for any well-behaving measure, but for the upper bound proofs the
above conditions are sufficient.

In practice, the fourth condition can be checked by differentiating
functions $f(N_X,cf_1N_X,m(A=a),n)$ and
$f(N_X,m(A=a)-cf_2(n-N_X),m(A=a),n)$ with respect to $N_x$. For
increasing $M$, the derivative $f'(N_X,cf_1N_x,m(A=a),n)$ should be
$f'>0$, when $cf_1>\frac{m(A=a)}{n}$, and the derivative
$f'(n-N_X,cf_2(n-N_X),m(A=a),n)$ should be $f'<0$, when
$cf_2>\frac{m(A=a)}{n}$. For decreasing $M$, the signs are reversed.

As an example, we show in Appendix \ref{appA} that the
$\chi^2$-measure, mutual information, two versions of the  
the $z$-score, and the $J$-measure are well-behaving. 

\section{Possible frequency values}

Before we go to the theoretical results, we introduce a graphical
representation, which simplifies the proofs.

Let us now consider the set of all legal frequency values, when the
data size $n$ and consequent $A=a$ (corresponding frequency $m(A=a)$)
are fixed.  All legal values of $N_X$ and $N_{XA}$ can be represented
in a two-dimensional space span by variables $N_X$ and
$N_{XA}$. Figure \ref{depproof1} shows a graphical representation of
the space.

\begin{figure}[!h]
\begin{center}
\includegraphics[width=\textwidth]{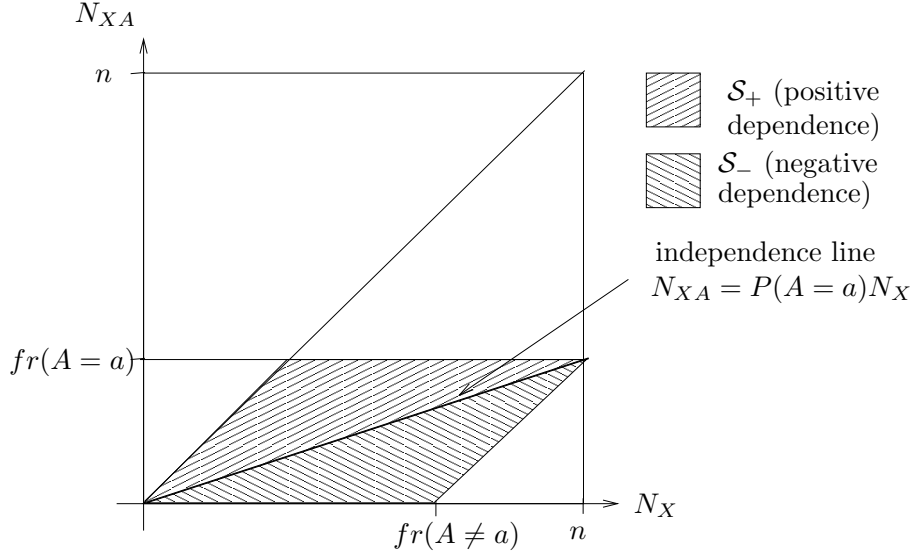}
\caption{Two-dimensional space of absolute frequencies $N_X$ and $N_{XA}$, 
when $A=a$ is fixed. In a given data set of size $n$, all points 
$(m(X),m(XA=a))$ lie in patterned areas.}
\label{depproof1}
\end{center}
\end{figure}

In any data set of size $N=n$, all possible frequency combinations
$N_X,N_{XA}$, $N_{XA}\leq N_X$, must lie in the triangle
$\{(0,0),(n,n),(n,0)\}$. If also the consequent $A=a$ is fixed, with
absolute frequency $m(A=a)$, the area of possible combinations is
restricted to the patterned areas in Figure \ref{depproof1} . Boundary
line $[(0,m(A=a)),(n,m(A=a))]$ follows from the fact that $m(XA=a)\leq
m(A=a)$ and line $[(m(A\neq a),0),(n,m(A))]$ from the fact that
$m(A\neq a)\geq m(XA\neq a)\Leftrightarrow m(XA=a)\geq m(X)-m(A\neq
a)$.  Line $[(0,0),(n,m(A=a))]$ is called the {\em independence line},
because on that line $N_{XA}=P(A=a)N_X$, i.e.\ $m(XA=a)=P(A=a)m(X)$,
and the corresponding $X$ and $A=a$ are statistically independent. If
the point lies above the independence line, the dependency is positive
($m(XA=a)>P(A=a)m(X)$), and below the line, it is negative
($m(XA=a)<P(A=a)m(X)$).

\begin{figure}[!h]
\begin{center}
\includegraphics[width=0.8\textwidth]{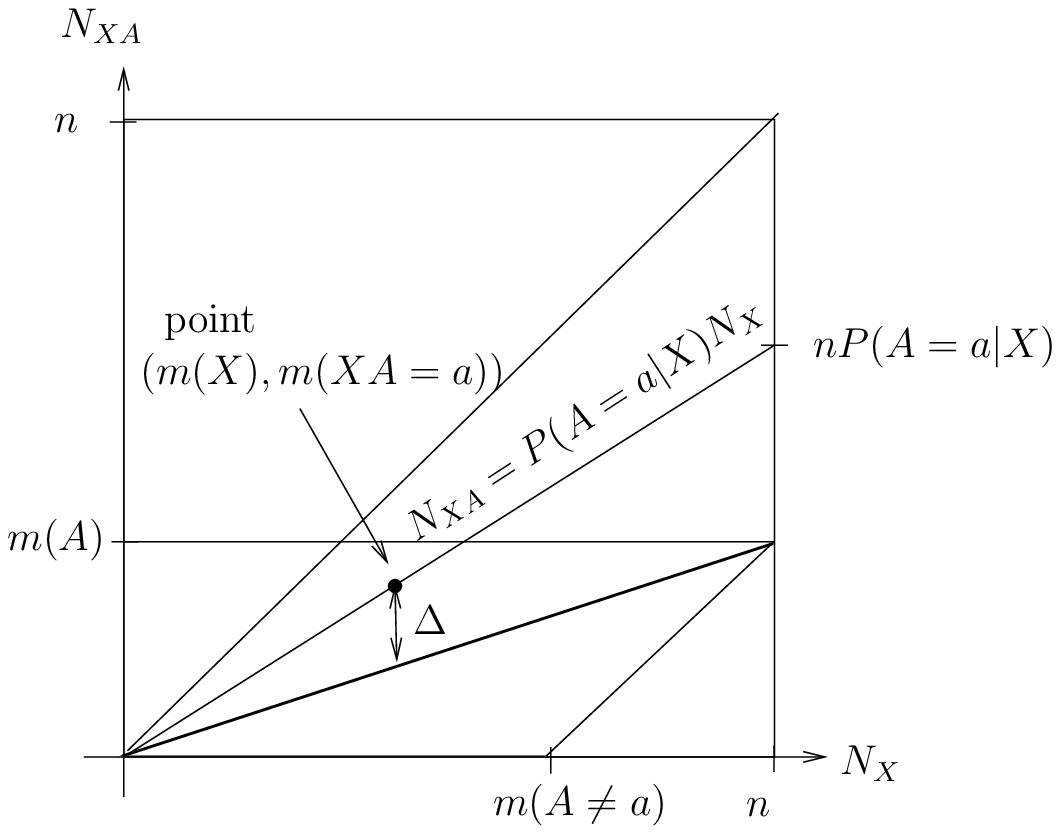}
\caption{Point $(m(X),m(XA=a))$ corresponding to rule $X \rightarrow A=a$. 
The vertical difference $\Delta$ from the independence line measures 
absolute leverage, $\Delta=n\delta$.}
\label{depproof2}
\end{center}
\end{figure}

Figure \ref{depproof2} shows a point $(m(X),m(XA=a))$ corresponding to
rule $X \rightarrow A=a$. In this case, the dependency is positive,
because $(m(X),m(XA=a))$ lies above the independence line. The slope
of line $[(0,0),(n,nP(A=a|X))]$ is the rule confidence
$P(A=a|X)=\frac{m(XA=a)}{m(X)}$. The vertical difference between point
$(m(X),m(XA=a))$ and the independence line, marked as $\Delta$,
defines the absolute leverage $\Delta=n\delta$. If the dependency is
negative and the point lies below the independence line, the leverage
is negative.

\begin{figure}[!h]
\begin{center}
\includegraphics[width=\textwidth]{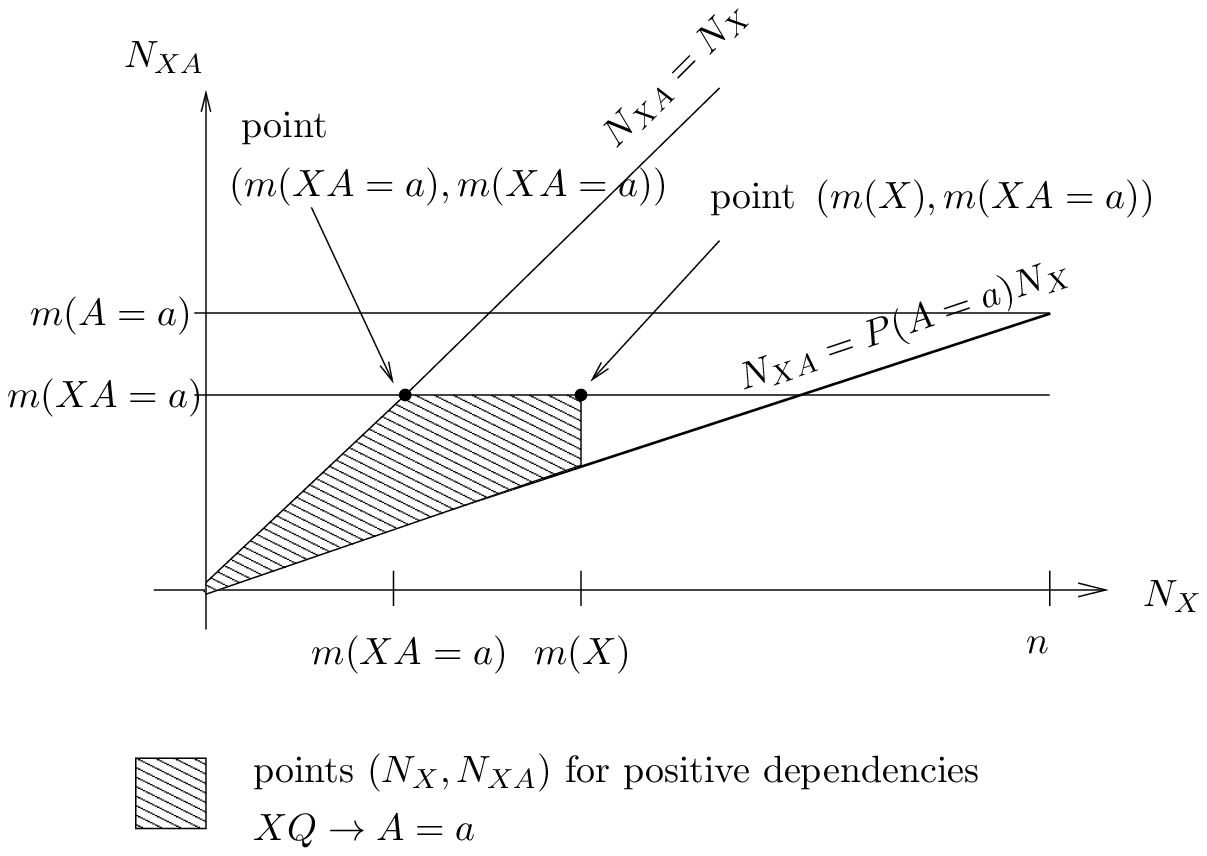}
\caption{When a point $(m(X),m(XA=a)$, corresponding to rule 
$X \rightarrow A=a$, is given, the points associated to more specific 
positive dependency rules $XQ \rightarrow A=a$ lie in the patterned area.}
\label{depproof3}
\end{center}
\end{figure}

Figure \ref{depproof3} shows how the knowledge on a rule $X
\rightarrow A=a$ can be utilized to determine the possible frequency
values of more specific positive dependency rules $XQ \rightarrow
A=a$. Since $m(XQA=a)\leq m(XA=a)$, all points $(m(XQ),m(XQA=a))$ must
lie under the line $[(0,m(XA=a)),(n,m(XA=a))]$. Because the
dependencies are positive, they also have to lie above the
independence line. In the next section, we will show that for a
well-behaving goodness measure $M$, point $(m(XA=a),m(XA=a))$
defines an upper bound (or lower bound) for any positive dependency
rule $XQ \rightarrow A=a$.

\begin{figure}[!h]
\begin{center}
\includegraphics[width=\textwidth]{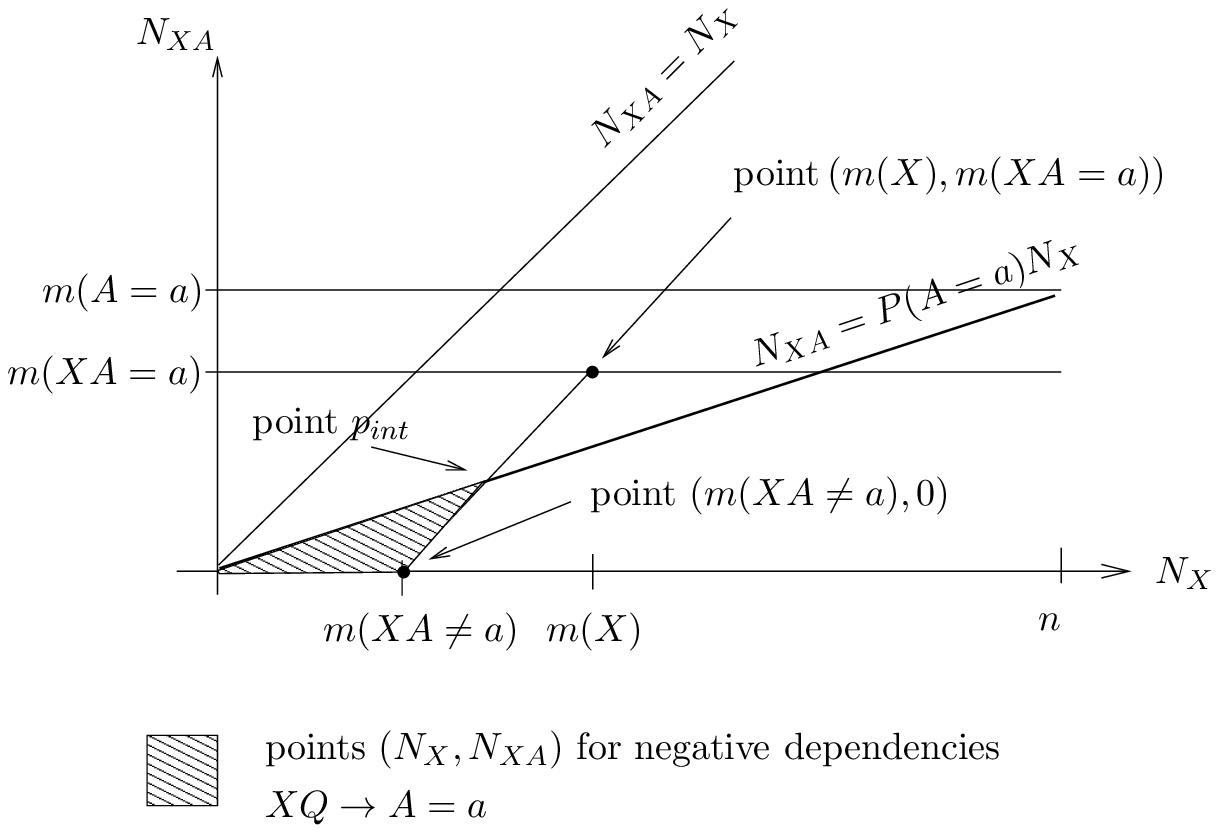}
\caption{When a point $(m(X),m(XA=a))$, corresponding to rule 
$X \rightarrow A=a$, is given, the points associated to more specific 
negative dependency rules $XQ \rightarrow A=a$ lie in the patterned area.}
\label{depproof4}
\end{center}
\end{figure}

Figure \ref{depproof4} shows the area, where all possible points for
negative dependency rules $XQ \rightarrow A=a$ must lie. Once again,
$m(XQA=a)\leq m(XA=a)$, and because the dependence is negative, the
points must lie under the independence line. In addition, the points
are restricted by line $[(m(X),m(XA=a)),(m(XA\neq a,0))]$. The reason
is that $m(XQA=a)=m(XA=a)-m(X\neg QA=a)$, where $m(X\neg QA=a)\in
[0,m(XA=a)]$.  On the other hand, $m(XQ)=m(X)-m(X\neg QA=a)-m(X\neg Q
A\neq a)\leq m(X)-m(X\neg Q A=a)$. So, the line contains the maximal
possible values $N_X=m(XQ)$ and the corresponding $N_{XA}=m(XQA=a)$
for any $m(X\neg QA=a)\in [0,m(XA=a)]$. In point $(m(X),m(XA=a))$,
$m(X\neg QA=a)=0$, and in point $(m(XA\neq a),0)$, $m(X\neg Q
A=a)=m(XA=a)$. In the following, we will show that point $(m(XA\neq
a),0)$ defines an upper bound for the $M$ value of any negative
dependency rule $XQ \rightarrow A=a$.

Figures \ref{depproof6} and \ref{depproof7} show the four axioms
graphically. According to axioms (ii) and (iii), function $f$
increases, when it departures from the independence line either
horizontally or vertically (Figure \ref{depproof6}). According to
axiom (iv), $f$ increases on lines $N_{XA}=cf_1N_X$, when $cf_1>P(A=a)$,
and decreases on lines $m(A=a)-N_{XA}=cf_2(n-N_X)$, when $cf_2>P(A=a)$
(Figure \ref{depproof7}).

\begin{figure}[!h]
\begin{center}
\includegraphics[width=\textwidth]{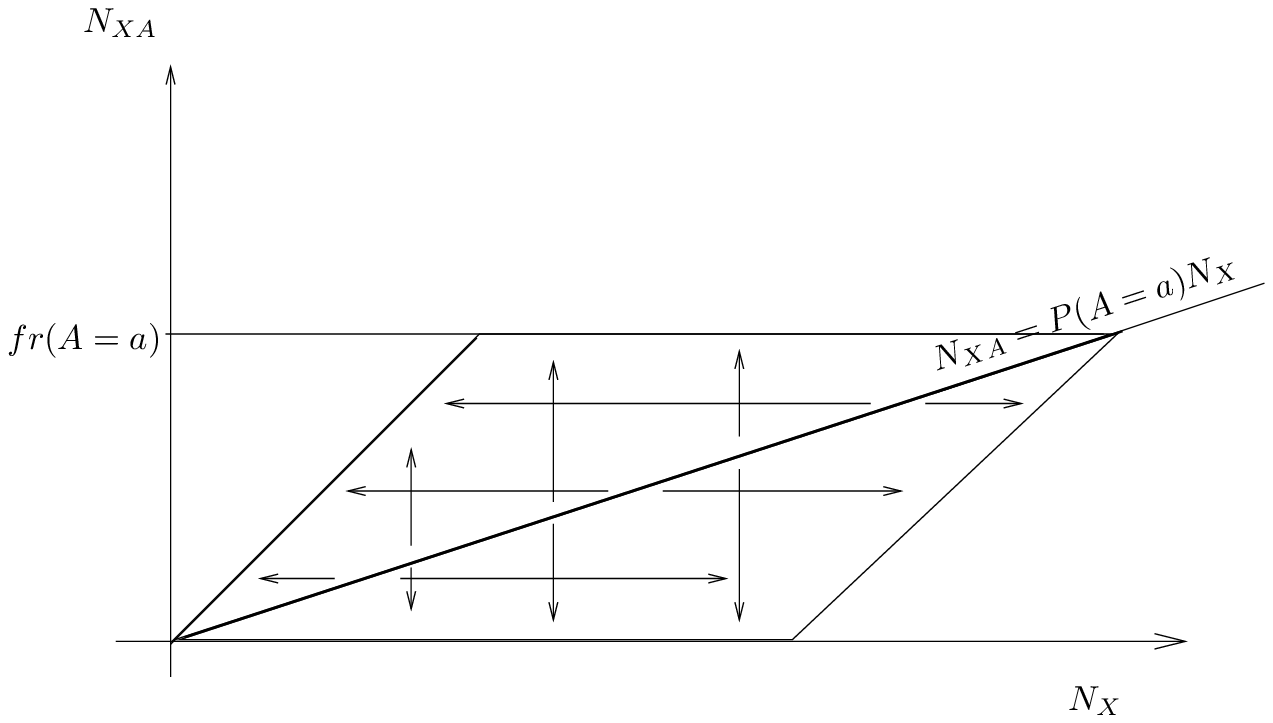}
\caption{Arrows show directions where $f$ increases (Axioms (ii) and (iii)).}
\label{depproof6}
\end{center}
\end{figure}

\begin{figure}[!h]
\begin{center}
\includegraphics[width=\textwidth]{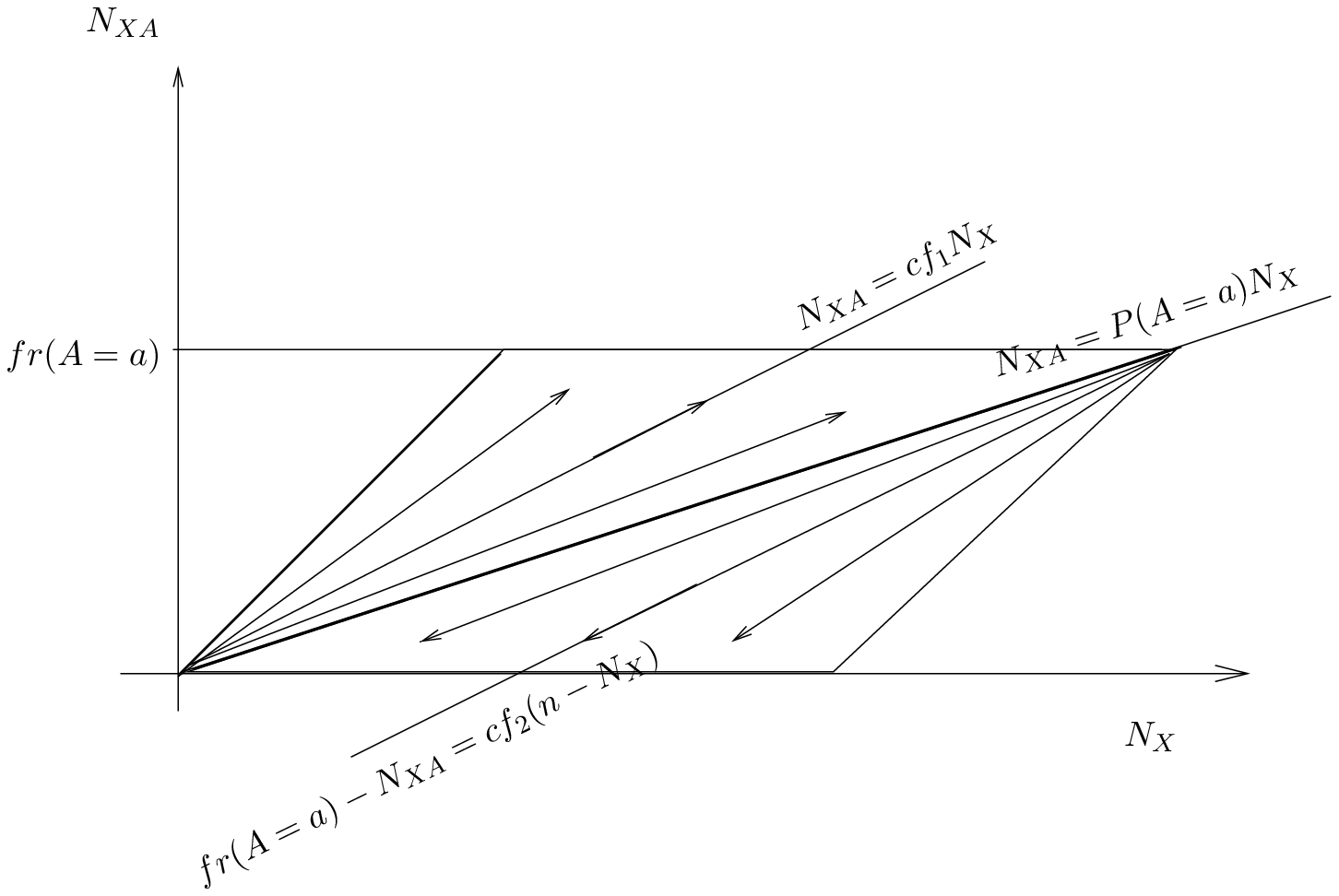}
\caption{Arrows show directions where $f$ increases (Axiom (iv)).}
\label{depproof7}
\end{center}
\end{figure}

\section{Useful upper bounds}

First we will note a couple of trivial properties, which follow from
the definition of well-behaving measures.

\begin{theorem}
\label{simple}
Let $M$ be a well-behaving, increasing measure. Let $S$ be a set of
legal values, as before. When $N=n$ and $N_A=m(A=a)$ are fixed, $M$ is
defined by function $f(N_X,N_{XA},m(A=a),n)$.

For positive dependencies hold   
\begin{itemize}
\item[(i)] $f$ gets its maximal values in set $S$  
on the border defined by points $(0,0)$, $(m(A=a),m(A=a))$, $(n,m(A=a))$.
\item[(ii)] $f$ gets its supremum (globally maximal value) in set $S$ in point 
$(m(A=a),m(A=a))$. 
\end{itemize}

and for negative dependencies hold 

\begin{itemize}
\item[(i)] $f$ gets its maximal values in set $S$  
on the border defined by points $(0,0)$, $(m(A\neq a),0)$, $(n,m(A=a))$.
\item[(ii)] $f$ gets its supremum in set $S$ in point 
$(m(A\neq a),0)$. 
\end{itemize}
\end{theorem}

\begin{proof}
Let us first consider positive dependencies.

(i) Since $M$ is well-behaving, $f$ is an increasing function of
$\delta$, when $\delta\geq 0$, in any point $N_X$. Therefore, it gets
its maximum value on the mentioned border. (ii) When $N_{XA}=m(XA=a)$
is fixed, $M$ being well-behaving, $f$ is a decreasing function of
$N_X$, when $N_X\leq \frac{n\cdot m(XA=a)}{m(A=a)}$. Therefore, $f$ is
decreasing on line $[(m(A=a),m(A=a)),(n,m(A=a))]$. On the other hand, we
know that for well-behaving $M$, $f(N_X,cfN_X,m(A=a),n)$ is increasing
with $N_X$, when $cf>\frac{m(A=a)}{n}$. When $cf=1$,
$f(N_X,cfN_X,m(A=a),n)$ coincides line
$[(0,0),(m(A=a),m(A=a))]$. Therefore, $f$ gets its maximum value, when
$N_X=N_{XA}=m(A=a)$.

For negative dependencies, the proof is similar. The only notable
exception is that now $f$ is increasing on line $[(0,0),(m(A\neq
  a),0)]$ and decreasing on line $[(m(A\neq a),0),(n,m(A=a))]$.
\end{proof}

This result can already be used for pruning in two ways. In the
beginning, some of the possible consequents $A=a$ may be pruned
out. Given a minimum threshold $min_M$, $A=a$ cannot occur in the
consequent of any sufficiently good positive dependency rule, if
$f(m(A=a),m(A=a),m(A=a),n)<min_M$. Similarly, $A=a$ cannot occur in
the consequent of any sufficiently good negative dependency rule, if
$f(m(A\neq a),0,m(A=a),n)<min_M$. We note that attribute $A$ can still
occur in the antecedent of good rules. This pruning property is
effective only with measures (like the mutual information), whose
supremum depends on $m(A=a)$. For example, with the $\chi^2$ measure,
the supremum is the same for all $m(A=a)$, and no pruning is possible,
when nothing else is known.

The second case occurs, when only $m(X)$ is known, but $m(XA=a)$ is
unknown. Now we can estimate an upper bound for both $X\rightarrow
A=a$ and all its specializations $XQ\rightarrow A=a$, by substituting
the best possible value for $N_{XA}$. In the case of positive
dependencies, the best possible value for $N_{XA}$ is
$\min\{m(X),m(A=a)\}$, and in the case of negative dependencies, it is
$\max\{0,m(X)-m(A\neq a)\}$. In practice, this means that when
$m(X)<m(A)$, the best possible $M$ value for positive dependence (in
point $(m(A=a),m(A=a))$ cannot be achieved any more, and the effective
pruning can begin. In the case of negative dependencies, the same
happens, when $m(X)$ becomes $m(X)<m(A\neq a)$, and point $(m(A\neq
a),0)$ is no more reachable.

The next theorem gives an upper bound for any positive or negative
dependency rule, when a more general rule is already known.

\begin{theorem}
Let $n$ and $m(A=a)$ be fixed and $M$, $S$ and $f$ like before.  Given
$m(X)$ and $m(XA=a)$ and an arbitrary attribute set $Q\subseteq
R\setminus (X\setminus\cup\{A\})$
\begin{itemize}
\item[(a)] for positive dependency $XQ \rightarrow A=a$ holds 
$f(m(XQ),m(XQA=a),m(A=a),n)\leq f(m(XA=a),m(XA=a),m(A=a),n)$, and 
\item[(b)] for negative dependency $XQ \rightarrow A=a$ holds 
$f(m(XQ),m(XQA=a),m(A=a),n)\leq f(m(XA\neq a),0,m(A=a),n)$.  
\end{itemize}
\end{theorem}

\begin{proof}
\begin{itemize}

\item[(a)] (Positive dependence) Figure \ref{depproof3} shows the area, 
where possible points $(m(XQ),m(XQA=a))$ for positive dependence can lie. 
With any $N_X\leq m(X)$, the maximum is achieved on the border defined by 
points $(0,0)$, $(m(XA=a),m(XA=a))$, and $(m(X),m(XA=a))$ ($\delta$ is 
maximal). On line $[(0,0),(m(XA=a),m(XA=a))]$ $f$ is increasing and on line 
$[(m(XA=a),m(XA=a)),(m(X),m(XA=a))]$ it is decreasing. Therefore, the 
global maximum is achieved in point $(m(XA=a),m(XA=a))$.

\item[(b)] (Negative dependence) Figure \ref{depproof4} shows the
  area, where possible points\\$(m(XQ),m(XQA=a))$ for negative
  dependence can lie. Once again, $f$ gets its maximal value for any
  $N_X$, when $-\delta$ is maximal. Therefore, the maximum must lie on
  the border defined by points $(0,0)$, $(m(XA\neq a),0)$, and the 
intersection point $p_{int}$. On line $[(0,0),(m(XA\neq a),0))]$, $f$ is 
increasing, and the maximal value is achieved in point $(m(XA\neq a),0)$. 
Therefore, it suffices to show that $f$ gets its maximum on line 
$[p_{int},(m(XA\neq a),0)]$ in the same point, $(m(XA\neq a),0)$. 

Figure \ref{depproof5} shows the proof idea. For every point $p_1$ on
line $[p_{int},(m(XA\neq a),0)]$, we can define a line of form
$m(A=a)-N_{XA}=cf_2(n-N_X)$, which goes through $p_1$. Because $p_1$
is under the independence line, $cf_2>P(A=a)$, and line
$m(A=a)-N_{XA}=cf_2(n-N_X)$ intersects $N_X$-axis in some point $p_2$
in the interval $[(0,0),(m(XA\neq a),0)]$. According to the definition
of well-behaving measures, $f$ is decreasing on line
$m(A=a)-N_{XA}=cf_2(n-N_X)$, and therefore $f$ gets a better value in
point $p_2$ than in point $p_1$. On the other hand, we already know
that $f$ gets a better value in $(m(XA\neq a),0)$ than any point
$p_2$. Therefore, $f$ must get its upper bound in point $(m(XA\neq
a),0)$.
\end{itemize}
\end{proof}

\begin{figure}[!h]
\begin{center}
\includegraphics[width=\textwidth]{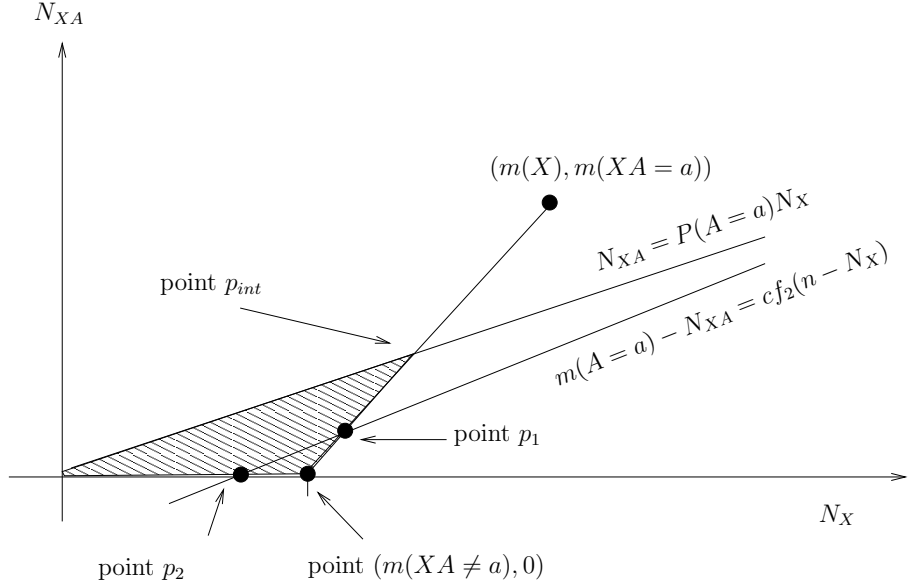}
\caption{Proof idea. Function $f$ is better in $p_2$ than in $p_1$ and 
better in point $(m(XA\neq a),0)$ than in $p_2$.}
\label{depproof5}
\end{center}
\end{figure}

These upper bounds enable more effective pruning than Theorem
\ref{simple}, because now pruning is possible even if $m(X)>m(A=a)$ or
$m(X)>m(A\neq a)$. The upper bounds are also tight in the sense that there
can be rules $XQ \rightarrow A$, which reach the upper bound values.

\section{Conclusions}

We have formalized the classical axioms for proper goodness measures
and extended them to to cover both positive and negative dependency
rules. We have shown that all such well-behaving goodness measures
achieve their upper bounds in the same points of the search
space. This is an important results, because it means that the same
generic search algorithm can be applied for a large variety of
commonly used goodness measures.

\appendix

\section{Example proofs for the good behaviour of common goodness measures}
\label{appA}

In the following, we show that the $\chi^2$-measure, mutual
information, two versions of the $z$-score
(e.g.\ \cite{statapr,lallich,bruzzese}), and the $J$-measure
\cite{smythgoodman} are well-behaving measures.  The first two
measures are defined for both positive and negative dependencies, while
the last three are defined only for positive dependencies. 

\begin{theorem}
Let $S\subseteq N^4$ be defined by constraints $0<N$, $0<N_X<N$, 
$0<N_A<N$, and $0\leq N_{XA}\leq \min\{N_X,N_A\}$. Measure $M$ is 
well-behaving, if it is defined by function
\begin{itemize}
\item[(a)] $\chi^2(N_X,N_{XA},N_A,N)=\frac{N(NN_{XA}-N_XN_A)^2}{N_X(N-N_X)N_A(N-N_A)}$,

\item[(b)] 
$MI(N_X,N_{XA},N_A,N)=N_{XA}\log \frac{N\cdot N_{XA}}{N_XN_A}+(N_X-N_{XA})\log \frac{N\cdot (N_X-N_{XA})}{N_X(N-N_A)}$\\
$+(N_A-N_{XA})\log \frac{N\cdot (N_A-N_{XA})}{(N-N_X)N_A}+(N-N_X-N_A+N_{XA})\log \frac{N\cdot (N-N_X-N_A+N_{XA})}{(N-N_X)(N-N_A)}$,

\item[(c)]$z_1(N_X,N_{XA},N_A,N)=\frac{\sqrt{N}(NN_{XA}-N_XN_A)}{\sqrt{N_XN_A(N^2-N_XN_A)}},$ when $NN_{XA}>N_XN_A$, and 0, otherwise, 

\item[(d)] $z_2(N_X,N_{XA},N_A,N)=\frac{NN_{XA}-N_XN_A}{\sqrt{N_XN_A(N-N_A)}},$ 
when $NN_{XA}>N_XN_A$, and 0, otherwise, and 

\item[(e)] $J(N_X,N_{XA},N_A,N)=N_{XA}\log(\frac{N_{XA}}{N_A})+(N_X-N_{XA})\log(\frac{N_X-N_{XA}}{N-N_A})-N_X\log(\frac{N_X}{N})$, when $NN_{XA}>N_XN_A$, and 
0, otherwise.
\end{itemize}
\end{theorem}

\begin{proof}
In the proofs, we assume that $N=n$ is fixed. We will simplify the functions 
by substituting $P(XA)=\frac{N_{XA}}{N}$, $P(X)=\frac{N_X}{N}$, 
$P(A)=\frac{N_A}{N}$. 
\begin{itemize}
\item[(a)] Conditions (i) and (ii): For $\chi^2$ the alternative expression is 
$$f_2(P(X),\delta,P(A),n)=\frac{n\delta^2}{P(X)(1-P(X))P(A)(1-P(A))}.$$ 
The derivative with respect to $\delta$ is  
$$f_2'=\frac{2n\delta}{P(X)(1-P(X))P(A)(1-P(A))},$$ which satisfies 
the conditions (i) and (ii). 

Condition (iii): When $P(XA)$, $P(A)$, and $n$ are fixed, $f$ can be 
expressed as 
$$f(P(X),P(XA),P(A),n)=\frac{n}{P(A)(1-P(A))}g(P(X)),$$ where
$g(P(X))=\frac{(P(XA)-P(X)P(A))^2}{P(X)(1-P(X))}$. The first factor is
constant, and therefore it is sufficient to differentiate $g(P(X))$ with
respect to
$P(X)$. 
$$g'(P(X))=\frac{-2P(XA)P(X)^2P(A)+P(X)^2P(A)^2-P(XA)^2+2P(X)P(XA)^2}{P(X)^2(1-P(X))^2}.$$ 
The denominator is
$[P(XA)-P(X)P(A)][-P(X)P(A)-P(XA)+2P(X)P(XA)]=[P(XA)-P(X)P(A)][P(X)(P(XA)-P(A))+P(XA)(P(X)-1)]$. The
first factor is leverage $\delta$ and the second factor is always
negative. Therefore, $g'<0$, when $\delta>0$, and $g'>0$, when
$\delta<0$.

Condition (iv): Let as first check the case, where
$N_{XA}=cf_1N_X$. Now $f$ becomes
$$f(P(X),cf_1P(X),P(A),n)=\frac{nP(X)(cf_1-P(A))^2}{P(A)(1-P(A))(1-P(X))}.$$ 
This is clearly an
increasing function of $P(X)$, when $cf_1\neq P(A)$. Let us then
check case $N_{XA}=m(A)-cf_2(n-N_X)$. Now $f$ becomes
$$f(P(X),P(A)-cf_2(1-P(X)),P(A),n)=\frac{n(1-P(X))(P(A)-cf_2)^2}{P(X)P(A)(1-P(A))}.$$ This is clearly a
decreasing function of $P(X)$, when $cf_2\neq P(A)$.

\item[(b)] In mutual information, the base of the logarithm is not
  defined, but usually it is assumed to be 2. However, transformation
  to the natural logarithm causes only an extra term +1, which
  disappears in differentiation. Therefore we will use the natural logarithms
  for simplicity. We recall that the derivative of a term of form 
  $g(x)\ln(g(X))$ is $g'(x)\ln(g(x))+g'(x)$.

Condition (i) and (ii): 
$MI$ can be expressed as function $f_2$:
{\small
\begin{multline*}
f_2(P(X),\delta,P(A),n)=n[(P(X)P(A)+\delta)\ln(P(X)P(A)+\delta)+
(P(X)(1-P(A))-\delta)\ln (P(X)\\(1-P(A))-\delta)+ ((1-P(X))P(A)-\delta)\ln
  ((1-P(X))P(A)-\delta) + ((1-P(X))(1-P(A))+\delta)\\\ln
  ((1-P(X))(1-P(A))+\delta)-P(A)\ln(P(A))
  -(1-P(A))\ln(1-P(A))-P(X)\ln(P(X))\\-(1-P(X))\ln(1-P(X))].
\end{multline*}}
The derivative of $f_2$ with respect to $\delta$ is
$$f_2'=n \ln \left(\frac{(P(X)P(A)+\delta)((1-P(X))(1-P(A))+\delta)}{((P(X)(1-P(A))-\delta)((1-P(X))P(A)-\delta)}\right).$$

This is the same as $n$ times the logarithm of the odds ratio $odds$, for
which holds $odds=1$, when $\delta=0$, $odds>1$, when $\delta>0$, and
$odds<1$, when $\delta<0$. Therefore, the logarithm is zero, when
$\delta=0$, negative, when $\delta<0$, and positive, when
$\delta>0$. 

Condition (iii): When $P(XA)$, $P(A)$, and $n$ are fixed, $f$ can be 
expressed as 
{\small
\begin{multline*}
g(P(X))=n[P(XA)\ln(P(XA))+(P(X)-P(XA))\ln(P(X)-P(XA)) +
  (P(A)-P(XA))\\ \ln (P(A)-P(XA))+(1-P(X)-P(A)+P(XA))\ln
  (1-P(X)-P(A)+P(XA))-P(A)\\ \ln(P(A))
  -(1-P(A))\ln(1-P(A))-P(X)\ln(P(X))-(1-P(X))\ln(1-P(X))].
\end{multline*}}
The derivative of $g$ with respect to $P(X)$ is
$$g'=n\ln \left(\frac{(P(X)-P(XA))(1-P(X))}{(1-P(X)-P(A)+P(XA))P(X)}\right).$$

Since $q=\frac{(P(X)-P(XA))(1-P(X))}{(1-P(X)-P(A)+P(XA))P(X)}=1$, when
$P(XA)=P(X)P(A)$, $q>1$, when $P(XA)<P(X)P(A)$, and $q<1$, when
$P(XA)>P(X)P(A)$, the logarithm is zero, when $X$ and $A=a$ are
independent, negative, when $X$ and $A=a$ are positively dependent,
and positive, when $X$ and $A=a$ are negatively dependent.

Condition (iv): When $N_{XA}=cf_1N_X$, $f$ becomes
{\small
\begin{multline*}
f(P(X),cf_1P(X),P(A),n)=n[P(X)cf_1\ln(P(X)cf_1)+P(X)(1-cf_1)\ln(P(X)(1-cf_1))\\+(P(A)-P(X)cf_1)\ln(P(A)-P(X)cf_1)+(1-P(X)-P(A)+
  P(X)cf_1)\ln(1-P(X)-P(A)\\+P(X)cf_1)-P(X)\ln
  P(X)-(1-P(X))\ln(1-P(X))-P(A)\ln P(A)-(1-P(A))\ln(1-P(A))].
\end{multline*}}
The derivative of $f$ is 
\begin{multline*}
f'=n[cf_1\ln(P(X)cf_1)+(1-cf_1)\ln(P(X)(1-cf_1))-cf_1\ln(P(A)-P(X)cf_1)-(1-cf_1)\\ \ln(1-P(X)-P(A)+P(X)cf_1)-\ln(P(X))+\ln(1-P(X))].
\end{multline*}
We should show that $f'>0$, when $cf_1>P(A)$. To find the lowest value of $f'$, 
we set $g(cf_1)=f'(P(X))$ and differentiate $g$ with respect to $cf_1$. 

\begin{align*}
g'(cf_1)&=n\left[\ln \left( \frac{P(X)cf_1(1-P(X)-P(A)+P(X)cf_1)}{P(X)(1-cf_1)(P(A)-P(X)cf_1)}\right)+\frac{P(X)cf_1}{P(A)-P(X)cf_1}\right.\\
&\left.-\frac{P(X)(1-cf_1)}{1-P(X)-P(A)+P(X)cf_1}\right]\\
&=n\left[\ln \left(\frac{P(XA)P(\neg X\neg A)}{P(X\neg A)P(\neg XA)}\right)+
\frac{P(XA)P(\neg X\neg A)-P(\neg XA)P(X\neg A)}{P(\neg XA)
P(X\neg A)}\right]\\
&=n\left[\ln(odds)+\frac{\delta}{P(\neg XA)P(X\neg A)}\right].
\end{align*}

The first term is the logarithm of the odds ratio, which is $>0$, when
$\delta>0$. So $g'>0$, when $\delta>0$. Therefore, $g$ is an
increasing function of $cf_1$ and gets its minimal value, when
$\delta=0$ and $cf_1=P(A)$. When we substitute this to $f'$, we get
$f'=\ln(1)=0$. This is the minimal value of $f'$, which is achieved
only, when $\delta=0$; otherwise $f'>0$, as desired.

Let us then check case $N_{XA}=m(A)-cf_2(n-N_X)$. Now
$P(XA)=P(A)-cf_2(1-P(X))$, $P(X\neg A)=P(X)(1-cf_2)-P(A)+cf_2$,
$P(\neg XA)=cf_2(1-P(X))$, $P(\neg X\neg A)=(1-P(X))(1-cf_2)$, and $f$
becomes
{\small
\begin{multline*}
f(P(X),P(A)-cf_2(1-P(X)),P(A),n)=\\
n[(P(A)-cf_2(1-P(X)))\ln (P(A)-cf_2(1-P(X)))+(P(X)(1-cf_2)-P(A)+cf_2)\ln (P(X)\\(1-cf_2)-P(A)+cf_2) +cf_2(1-P(X))\ln(cf_2(1-P(X)))+(1-P(X))(1-cf_2)\ln((1-P(X))\\(1-cf_2))-P(X)\ln P(X)-(1-P(X))\ln(1-P(X))-P(A)\ln(P(A))-(1-P(A))\ln(1-P(A))]\\
=n[(P(A)-cf_2(1-P(X)))\ln (P(A)-cf_2(1-P(X)))+(P(X)(1-cf_2)-P(A)+cf_2)\ln (P(X)\\(1-cf_2)-P(A)+cf_2) +cf_2(1-P(X))\ln(cf_2)+(1-P(X))(1-cf_2)\ln(1-cf_2)-P(X)\ln P(X)\\
-P(A)\ln(P(A))-(1-P(A))\ln(1-P(A))]
\end{multline*}}
The derivative of $f$ with respect to $P(X)$ is
\begin{multline*}
f'=n[cf_2\ln\frac{P(A)-cf_2(1-P(X))}{cf_2}+(1-cf_2)\ln \frac{P(X)(1-cf_2)-P(A)+cf_2}{1-cf_2}\\-\ln(P(X))].
\end{multline*}
We should show that $f'<0$, when $cf_2>P(A)$. To find the largest value of 
$f'$, we set $g(cf_2)=f'(P(X))$ and differentiate $g$ with respect to $cf_2$. 

\begin{align*}
g'(cf_2)&=n\left[\ln \left(\frac{P(A)-cf_2(1-P(X))}{cf_2}\right)-\ln \left(\frac{P(X)-P(A)+cf_2(1-P(X))}{1-cf_2}\right)\right.\\
&\left. -\frac{P(A)}{P(A)-cf_2(1-P(X))}+\frac{1-P(A)}{P(X)-P(A)+cf_2(1-P(X))}\right]\\
&=n\left[\ln \left(\frac{P(XA)(1-cf_2)}{cf_2P(X\neg A)}\right)-\frac{P(A)}{P(XA)}+\frac{1-P(A)}{P(X\neg A)}\right].
\end{align*}

Because the dependence is negative, $-P(XA)>-P(X)P(A)$ and
$P(X)(1-P(A))<P(X\neg A)$. Therefore, the sum of the last two terms is
$<0$. The first term becomes $\ln(odds)$, when we substitute
$cf_2=P(A|\neg X)$. For negative dependence, also $\ln(odds)<0$, and
therefore $g'<0$. Because $g$ is decreasing with $cf_2$, $f'$ gets its
maximum value, when $cf_2$ is minimal, i.e.\ $P(A)$. When we
substitute this to $f'$, it becomes $f'=0$. This is the maximal value
of $f'$, which is achieved only when $cf_2=P(A)$ and
$\delta=0$. Otherwise, $f'<0$, as desired.

\item[(c)] Now it is enough to check the conditions only for the
  positive dependence. 

Conditions (i) and (ii): $z_1$ can be expressed as
$$f_2(P(X),\delta,P(A),n)=\frac{\sqrt{n}\delta}{\sqrt{P(X)P(A)(1-P(X)P(A))}},$$
when $\delta>0$, and $f_2=0$, otherwise. This is clearly an
increasing function of $\delta$ and gets its minimum value $0$, when
$\delta\leq 0$.

Condition (iii): When $P(XA)$, $P(A)$, and $n$ are fixed, $f$ can be 
expressed as 
$$f(P(X),P(XA),P(A),n)=\frac{\sqrt{n}}{\sqrt{P(A)}}g(P(X)),$$ where
$g(P(X))=\frac{P(XA)-P(X)P(A)}{\sqrt{P(X)(1-P(X)P(A))}}$. The derivative
of $g$ with respect to $P(X)$ is

\begin{align*}
g'&=\frac{-2P(X)P(A)(1-P(X)P(A))-(P(XA)-P(X)P(A))(1-2P(X)P(A))}{2(P(X)(1-P(X)P(A)))^{\frac{3}{2}}}\\
&=\frac{-P(X)P(A)(1-P(XA))-P(XA)(1-P(X)P(A))}{2(P(X)(1-P(X)P(A)))^{\frac{3}{2}}}.
\end{align*}

Because $P(X)P(A)<1$ and $P(XA)\leq 1$, $g'<0$ always.

Condition (iv): When $N_{XA}=cf_1N_X$, $f$
becomes $$f(P(X),cf_1P(X),P(A),n)=\frac{\sqrt{nP(X)}(cf_1-P(A))}{\sqrt{P(A)(1-P(X)P(A)})},$$
which is clearly an increasing function of $P(X)$, when $cf_1\geq
P(A)$.

\item[(d)] Conditions (i) and (ii): $z_2$ can be expressed as 
$$f_2(P(X),\delta,P(A),n)=\frac{\sqrt{n}\delta}{\sqrt{P(X)P(A)(1-P(A))}},$$
when $\delta>0$, and $f_2=0$, otherwise. This is clearly an increasing 
function of $\delta$ and gets its minimum value 
$0$, when $\delta\leq 0$.

Condition (iii): When $P(XA)$, $P(A)$, and $n$ are fixed, $f$ can be 
expressed as 
$$f(P(X),P(XA),P(A),n)=\frac{\sqrt{n}}{\sqrt{P(A)(1-P(A))}}g(P(X)),$$ where
$g(P(X))=\frac{P(XA)-P(X)P(A)}{\sqrt{P(X)}}$. The derivative
of $g$ with respect to $P(X)$ is
$$g'(P(X))=\frac{-P(X)P(A)-P(XA)}{2P(X)^{\frac{3}{2}}},$$ which is $<0$ always. 

Condition (iv): When $P(XA)=cf_1P(X)$, $f$ becomes
$$f(P(X),cf_1P(X),P(A),n)=\frac{\sqrt{nP(X)}(cf_1-P(A))}{\sqrt{P(A)(1-P(A))}},$$ which is clearly an increasing function of $P(X)$, when $cf_1\geq P(A)$.

\item[(e)] Like in the mutual information, we will use the natural 
logarithm for simplicity. 

Conditions (i) and (ii): $J$ can be expressed as 
\begin{multline*}
f_2(P(X),\delta,P(A),n)=n\left[(P(X)P(A)+\delta)\ln\left(\frac{P(X)P(A)+\delta}{P(A)}\right)\right.\\
\left.+(P(X)(1-P(A))-\delta)\ln\left(\frac{P(X)(1-P(A))-\delta}{(1-P(A))}\right)-P(X)\ln(P(X))\right].
\end{multline*}

The derivative of $f_2$ with respect to $\delta$ is
\begin{align*}
f_2'&=n\left[\ln \left(\frac{P(X)P(A)+\delta}{P(A)}\right)-\ln \left(\frac{P(X)(1-P(A))-\delta}{(1-P(A))}\right)-\ln(P(X))\right]\\
&=n\ln \left(\frac{(P(X)P(A)+\delta)(1-P(A))}{(P(X)(1-P(A))-\delta)P(X)P(A)}\right ).
\end{align*}

The argument of the logarithm is 1, if $\delta=0$, and otherwise it is
$>1$. Therefore, $f_2'\geq 0$, when $\delta\geq 0$.

Condition (iii): When $P(XA)$, $P(A)$, and $n$ are fixed, $f$ can be 
expressed as 
\begin{multline*}
g(P(X))=n[P(XA)\ln \left(\frac{P(XA)}{P(A)}\right)+(P(X)-P(XA))\ln \left(\frac{P(X)-P(XA)}{(1-P(A))}\right)\\-P(X)\ln(P(X))].
\end{multline*}

The derivative of $g$ with respect to $P(X)$
is 
$$g'=n[\ln\left(\frac{P(X)-P(XA)}{(1-P(A))}\right)-\ln(P(X))]=n\ln \left(\frac{P(X)-P(XA)}{P(X)(1-P(A))}\right).$$ 
The argument of the logarithm is $<1$ and thus $g'<0$, if there is a
positive dependency.

Condition (iv): When $N_{XA}=cf_1N_X$, $f$ becomes
\begin{multline*}
f(P(X),cf_1P(X),P(A),n)=\\n\left[P(X)cf_1\ln \left(\frac{P(X)cf_1}{P(A)}\right)+P(X)(1-cf_1)\ln \left(\frac{P(X)(1-cf_1)}{(1-P(A))}\right)-P(X)\ln(P(X))\right]. 
\end{multline*}

The derivative of $f$ with respect to $P(X)$ is
\begin{align*}
f'&=n\left[cf_1\ln\left(\frac{P(X)cf_1}{P(A)}\right)+(1-cf_1)\ln\left(\frac{P(X)(1-cf_1)}{(1-P(A))}\right)-\ln(P(X))\right]\\
&=n\left[cf_1\ln\left(\frac{cf_1(1-P(A))}{(1-cf_1)P(A)}\right)+\ln\left(\frac{1-cf_1}{(1-P(A))}\right)\right].
\end{align*}

We should show that $f'>0$, when $cf_1>P(A)$. To find the lowest value of $f'$, 
we set $g(cf_1)=f'(P(X))$ and derivative $g$ with respect to $cf_1$. 
$$g'(cf_1)=n\ln \left(\frac{cf_1(1-P(A))}{(1-cf_1)P(A)}\right).$$
Clearly $g'=0$, if $cf_1=P(A)$, and $g'>0$, if $cf_1>P(A)$. When we
substitute the minimum value $cf_1=P(A)$ to $f'$, we get
$f'=n[P(A)\ln(1)+\ln(1)]=0$. When $cf_1>P(A)$, $f'>0$, and $f$ is an
increasing function of $P(X)$.

\end{itemize}
\end{proof}

\bibliographystyle{plain}


\end{document}